\documentclass[11pt]{article}
\pdfoutput=1
\usepackage{natbib}
 \bibpunct{(}{)}{;}{a}{,}{,}
\usepackage{graphicx}
\usepackage{amsmath,amssymb,amsthm}
\usepackage{bm}
\usepackage{rotating}
\usepackage[lined,ruled]{algorithm2e}
\usepackage{multirow}
\usepackage{rotating}
\usepackage{multicol}
\usepackage{titlesec}
\usepackage{latexsym}
\usepackage[pdftex,colorlinks=true,linkcolor=blue,citecolor=blue,urlcolor=blue,bookmarks=false,pdfpagemode=None]{hyperref}
\usepackage{url}
\usepackage{verbatim}
\usepackage{fancyhdr}
\usepackage{setspace}
\usepackage{paralist}
\usepackage{lineno}
\usepackage[top=1in, bottom=1in, left=1in, right=1in]{geometry}
\usepackage{lscape}
\usepackage{dcolumn}

\newcolumntype{.}{D{.}{.}{-1}}

\newtheorem{definition}{Definition}
\newtheorem{remark}{Remark}
\newtheorem{example}{Example}

\newtheorem{theorem}{Theorem}
\newtheorem{proposition}{Proposition}
\newtheorem{assumption}{Assumption}

\newcommand{\iv}{I}

\newcommand{\VYA}{\textbf{Y}(\textbf{A})}
\newcommand{\VYB}{\textbf{Y}(\textbf{B})}
\newcommand{\VYZ}{\textbf{Y}(\textbf{Z})}
\newcommand{\VYZo}{\textbf{Y}(\textbf{Z}_0)}

\newcommand{\Z}{\textbf{Z}}
\newcommand{\calZ}{\mathcal{Z}}
\newcommand{\A}{\textbf{A}}
\newcommand{\B}{\textbf{B}}
\newcommand{\Ni}{\mathcal{N}}
\newcommand{\sexp}{\mathbb{E}_\mathcal{R}}
\newcommand{\sexpg}{\mathbb{E}_\mathcal{G}}
\newcommand{\svar}{\mathbb{V}_\mathcal{R}}

\newcommand{\sprob}{\mathbb{P}}
\newcommand{\scov}{\mathbb{C}ov_\mathcal{R}}
\newcommand{\G}{\mathcal{G}}
\newcommand{\R}{\mathcal{R}}
\newcommand{\Y}{\textbf{Y}}
\newcommand{\CRD}{Completely Randomized Design }
\newcommand{\BD}{Bernoulli Design }
\newcommand{\CBD}{Conditional Bernoulli Design }

\begin{document}
\pagestyle{empty}

\title{
Limitations of design-based causal inference and A/B testing under arbitrary and network interference
\protect\thanks{Guillaume W. Basse is a graduate student in the Department of Statistics at Harvard University (\href{mailto:gbasse@fas.harvard.edu}{gbasse@fas.harvard.edu}). Edoardo M.~Airoldi is an Associate Professor of Statistics at Harvard University (\href{mailto:airoldi@fas.harvard.edu}{airoldi@fas.harvard.edu}).
This work was partially supported 
 by the National Science Foundation under grants 
  CAREER IIS-1149662 and IIS-1409177,
 and by the Office of Naval Research under grants 
  YIP N00014-14-1-0485 and N00014-17-1-2131. 
 Guillaume W. Basse is a Google Fellow in Statistics.
}}

\author{
 Guillaume W. Basse, Edoardo M. Airoldi\\
 Department of Statistics\\
 Harvard University, Cambridge, MA 02138, USA}
\date{}

\maketitle
\thispagestyle{empty}

\newpage
\begin{abstract}
Randomized experiments on a network often involve interference between connected units; i.e., a situation in which an individual's treatment can affect the response of another individual. Current approaches to deal with interference, in theory and in practice, often make restrictive assumptions on its structure---for instance, assuming  that interference is local---even when using otherwise  nonparametric inference strategies. This reliance on explicit restrictions on the interference mechanism suggests a shared intuition that inference is impossible without any assumptions on the interference structure. In this paper, we begin by formalizing this intuition in the context of a classical nonparametric approach to inference, referred to as design-based inference of causal effects. Next, we show how, always in the context of design-based inference, even parametric structural assumptions that allow the existence of unbiased estimators, cannot guarantee a decreasing variance even in the large sample limit. This lack of concentration in large samples is often observed empirically, in randomized experiments in which interference of some form is expected to be present. This result has direct consequences for the design and analysis of large experiments---for instance, in online social platforms---where the belief is that large sample sizes automatically guarantee small variance. More broadly, our results suggest that although strategies for causal inference in the presence of interference borrow their formalism and main concepts from the traditional causal inference literature, much of the intuition from the no-interference case do not easily transfer to the interference setting. 
\vfill
\end{abstract}

\newpage
\singlespacing
\small
\tableofcontents
\normalsize
\doublespacing

\newpage
\pagestyle{plain}
\setcounter{page}{1}


\section{Introduction}

In modern randomized experiments, ``interference'' typically means that the response of a given unit
to a certain treatment may depend on the treatment assigned to other units.
In online marketing, \citep{aral2012identifying} showed that adoption of a product by an individual
in a social network tends to encourage the adoption of the same product by their neighbors. Similar
examples can be found in epidemiology, where vaccinating a certain percentage of the population
is expected to lower the health risk for the entire population \citep{hudgens2012toward}, or in 
political science where encouraging an individual to vote is expected to increase the turnout for 
other members of the household \citep{sinclair2012detecting}. Researchers have had some 
success in addressing the question by making assumptions on the interference mechanism 
\citep{hudgens2012toward, sinclair2012detecting, basse2016, ugander2013graph, bowers2016aa, bowers2016ab}. However, because these mechanisms
are often very complex, it is tempting to believe that we can rely on large sample sizes to avoid 
making assumptions. The goal of this paper is to show  how wrong this intuition is, precisely.


\subsection{Background and set-up}

%
%
%

In a randomized experiment (or A/B test), N units are randomly assigned to treatments A or B, 
and an outcome of interest $Y_i(\Z^{obs})$ is measured for each unit $i$, where 
$\Z^{obs} \in \calZ = \{A, B\}^N$ is the observed assignment vector. There are generally two
paradigms available for making inference: the model-based approach, 
and the design-based approach. 

In the model-based approach, 
the vector of observed outcomes is modeled $Y_i(\Z^{obs}) \sim F(\vec{\xi})$. The estimands
are generally functions of the parameters $\theta = g(\vec{\xi})$ and the estimators are usually
obtained using maximum likelihood or bayesian inference. Two key points are that the estimators are selected using the model $F(\vec{\xi})$, and the properties of the estimator are typically derived in some asymptotic regime and incorporate randomness due to both the model for the outcomes and the randomization distribution.

In contrast, the design-based approach we consider in this paper takes the opposite perspective and considers the 
\textit{potential outcomes} $\{\Y(\Z)\}_\Z$ as fixed but 
a-priori unknown quantities. The assignment mechanism $\R$, which assigns a probability
$\sprob_\R(\Z)$ to every vector $\Z \in \calZ$ provides the only source of randomness. Estimands
are then generally functions of the potential outcomes $\theta = g(\{\Y(\Z)\})$, and it is 
desirable to find a pair $(\hat{\theta}, \sprob_\R)$ such that the estimator $\hat{\theta}$ has good
properties under the design $\R$ in finite samples. 
One goal of this paper is to show that under arbitrary interference, there exist no pair $(\hat{\theta}, \sprob_\R)$ with good  properties. 
 
 A third approach to inference exists, typically referred to as \textit{model-assisted}, which a attempts to blend the two approaches by using a model for the outcomes to inform the choice of design or correction factors for estimators in a given family  \citep{sarndal2003model,basse2015optimal}. Since the actual inference is then performed from a design-based perspective, this
 approach also falls under the scope of our paper.

The idea that 
inference is not feasible without assumptions on the interference mechanism has been suggested in the literature for observational studies \citep[e.g., see][]{shalizi2011homophily}, and also in the context of randomized experiments; for instance \citet{aronow2013estimating}
state that under arbitrary interference, ``it is clear that there would be no meaningful way to use the results  of the experiment'', and then focuses on how to do inference under some form of restricted interference. In the first, and what is to our knowledge also the only, attempt at formalizing the issue, \citet{manski2013identification} shows that in a model-based setting, the distribution of potential outcomes is not identifiable under arbitrary interference. 
 
\subsection{Contributions}

Our paper makes three contributions. First, we extend and clarify the results of 
 \citep{manski2013identification} in the context of design-based inference (Section~\ref{section:all-interferences})
 and prove, among other things, that there exist no consistent estimators under arbitrary interference. Second,
 we extend the work of \citep{aronow2013estimating} by focusing on a popular design and an 
 interference structure commonly assumed in network settings. Assuming an Erdos-Renyi for the network, 
 we show that for a class of unbiased estimators, consistency depends on the parameter $p$ of the model
 for the graph. Finally, using the concept of \textit{effective treatment} introduced by 
 \citep{manski2013identification}, we provide analytical insights into the general problem of interference and the 
 convergence of estimators. 
 
 For researchers,  this paper is meant to provide a clear formalization of some intuitions that have   been suggested. For practitioners, this paper is meant to offer a convincing
 argument for the necessity of making explicit assumptions about the structure of the interference mechanism, to rely be able to rely on  asymptotic standard errors for causal inference.
 

\section{No-interference and arbitrary interference}
\label{section:all-interferences}

\subsection{Setup}

To avoid pathological cases, we assume that the potential outcomes are bounded:
\begin{assumption}[bounded outcomes]\label{bounded-outcomes}
	There exists $M$ such that for all $N$:
	\begin{equation*}
		0 < Y_i(\Z) < M \quad \forall \,\, i = 1,\ldots N
	\end{equation*}
\end{assumption}
%
Although many of our results could generalize to large classes of designs $\R$, we will focus on 
three assignment mechanisms for the sake of clarity. The Completely Randomized Design (CRD) 
assigns a fixed number $N_A$
of randomly selected units to treatment A, and the other $N_B = N-N_A$ units to treatment B. The Bernoulli 
Design (BD) assigns independently each unit to treatments A or B with probability 1/2, while the Conditional 
Bernoulli Design (CBD) operates similarly to the \BD, but excludes the assignments $\Z = \A$ and $\Z = \B$ 
in which all units are assigned to A, or all units assigned to B.

The estimands considered in this paper are of the form $\theta(\A,\B) = g(\VYA, \VYB)$, and estimators 
will be denoted by $\hat{\theta}(\Z)$. 

\begin{example}
	The most popular estimand in causal inference is the average total treatment effect \citep{imbens2015causal}:
	\begin{equation*}
		\theta(\A,\B) = \overline{Y}(\A) - \overline{Y}(\B)
	\end{equation*}
	but our results are stated in greater generality.
\end{example}

All statements about estimators must hold regardless of the value 
of the potential outcomes $\{\VYZ\}_\Z$. In particular we consider the following definition of unbiasedness:
\begin{definition}\label{def:unbiased}
	Let $\theta(\A,\B) = g(\VYA,\VYB)$ be the estimand of interest. An estimator 
	$\hat{\theta}(Z)$ is said to be unbiased if:
	\begin{equation*}
		\sexp[\hat{\theta}(\Z)] = \theta(\A,\B)
	\end{equation*}
	for all values of the potential outcomes $\{\VYZ\}_\Z$ (the expectation being taken with 
	respect to the assignment mechanisms).
\end{definition}
 In particular, an estimator is not unbiased if the equality $E(\hat{\theta}(\Z)) = \theta(\A,\B)$ only 
 holds for specific values of $\{\VYZ\}_Z$. Similar considerations apply to the concept of consistency.

\subsection{Inference under no-interference}
\label{section:sutva}

The most prominent applications involve situations in which interference is a nuisance
and can be safely be assumed away. The main clause of the popular Stable Unit Treatment 
Value Assumption (SUTVA; \cite{rubin1980comment}) implies that there is no interference 
between units:
\begin{assumption}[No interference]\label{no-interference}
	\begin{equation*}
		Y_i(\Z) = Y_i(Z_i) \quad \forall \,\, i=1, \ldots, N
	\end{equation*}
\end{assumption}
Studying the no-interference case is important because it often implicitly guides 
our intuition in more complex scenarios, especially the large sample behavior or estimators. One aim of 
that paper is to show the dangers of that intuition.

 We illustrate the standard characteristics of inference under this assumption by focusing on the \CRD and
the average total treatment effect:
\begin{equation*}
	\theta(\A,\B) = \frac{1}{N} \sum_i^N (Y_i(\A) - Y_i(\B))
\end{equation*}
as our estimand. The following estimator:
\begin{equation}\label{eq:basic-estimator}
	\hat{\theta}(\Z) = \frac{1}{N_A} \sum_i^N \iv(Z_i=A) Y_i(\Z) - \frac{1}{N_B} \sum_i^N \iv(Z_i=B) Y_i(\Z)
\end{equation}
which under Assumption~\ref{no-interference} simplifies to:
\begin{equation*}
	\hat{\theta}(\Z) = \frac{1}{N_A} \sum_i^N \iv(Z_i=A) Y_i(A_i) -  \frac{1}{N_B} \sum_i^N \iv(Z_i=B) Y_i(B_i)
\end{equation*}
can easily be shown to be unbiased for $\theta(\A,\B)$ under the \CRD. That is:
\begin{equation*}
	\sexp[\hat{\theta}(\Z)] = \theta(\A,\B)
\end{equation*}
and the variance is (see e.g \cite{imbens2015causal} chapter 6):
\begin{flalign*}
	\svar[\hat{\theta}] &= \frac{V_A}{N_A} + \frac{V_B}{N_B} -  \frac{V_\theta}{N}\\
	 &\leq \frac{1}{N_A}\frac{1}{N-1} \sum_i (Y_i(\A) - \overline{Y}(\A))^2 \\
	 &+ \frac{1}{N_B}\frac{1}{N-1} \sum_i (Y_i(\B) - \overline{Y}(\B))^2 \\
	&\leq \frac{1}{N_A(N-1)}NM^2 + \frac{1}{N_B(N-1)}NM^2\\
	&\leq \frac{4M^2}{N-1}
\end{flalign*}
where the terms $V_A$, $V_B$, and $V_\theta$ are defined in appendix. So under Assumption~\ref{bounded-outcomes} and Assumption~\ref{no-interference} there exists an
unbiased estimator, with variance of order $O(1/N)$.

\subsection{Arbitrary interference}
\label{section:arbitrary-interference}

Under arbitrary interference, the outcome for unit $i$ depends on the entire assignment vector $\Z$, not
just on its own assignment $Z_i$. We show in this section that not only is the estimator in 
Equation~\eqref{eq:basic-estimator} biased, but that unbiased estimators 
(in the sense of Definition~\ref{def:unbiased}) simply do not exist for a wide class of designs which 
includes both the \CRD and the \CBD
\begin{theorem}\label{th:bias}
	Consider any non-degenerate \footnote{This excludes estimands that are constant, i.e that don't depend on $\Y(\A)$ and $\Y(\B)$.} estimand $\theta(\A, \B)$, and any assignment mechanism $\R$
	such that $\sprob_\R(\Z = \A) = \sprob_\R(\Z=\B) = 0$. There exists no unbiased estimator of 
	$\theta(\A,\B)$ under $\R$. 
\end{theorem}
If the design assigns non zero probability to the treatment allocation vectors $\A$ and $\B$, then unbiased estimators 
exist for a restricted class of estimands. The following proposition considers the case of the \BD, as
a concrete example:
\begin{proposition}\label{prop:bias-extended}
	Consider the Bernoulli Design. If the estimand is of additive form 
	$\theta(\A,\B) = \theta_1(\A) + \theta_2(\B)$, then unbiased estimators are of the form 
	\begin{equation*}
		\hat{\theta}(\Z) = C(\Z) + 2^N \iv(\Z = \A)\theta_1(\A) + 2^N\iv(\Z = \B)\theta_2(\B)
	\end{equation*}
	where $C(\Z)$ does not depend on any potential outcomes 
	 \footnote{This term may incorporate covariates, as with \textit{model-assisted} 
	estimators. In the
	absence of external information, we focus on the case $C(\Z) = 0$.} and satisfies:
	\begin{equation*}
		\sum_{\Z\in\mathcal{Z}} C(\Z) = 0
	\end{equation*} 
	For other types of estimands $\theta(\A,\B)$, there exist no unbiased estimators.
\end{proposition}
Taken together, Theorem~\ref{th:bias} and Proposition~\ref{prop:bias-extended} formalize a 
simple idea: if $\Z \neq \Z'$, then $\Y(\Z')$ is completely non-informative for $\Y(\Z)$. In particular,
the only assignments which, if observed, could provide information about the estimand $\theta(\A, \B)$
are $\Z=\A$ and $\Z=\B$.  The estimator of Proposition~\ref{prop:bias-extended} reflects this by 
evaluating to 0 for any uninformative assignment (that is, $\Z \not \in \{\A, \B\}$), and assigning large
weight to the only two informative assignments ($\Z=\A$ and $\Z=\B$). This is a known case of 
failure for this kind of Horvitz-Thompson estimators \citep{basu2011essay}.
The following example shows the implications of 
Theorem~\ref{th:bias} and Proposition~\ref{prop:bias-extended} when the estimand is the average 
total treatment effect:
\begin{example}
Consider the following estimand:
\begin{equation*}
	\theta(\A, \B) = \overline{Y}(\A) - \overline{Y}(\B)
\end{equation*}	
and let $0 < N_1 < N$. $\theta(\A,\B)$ has the additive form of 
Proposition~\ref{prop:bias-extended} with $\theta_1(\A) = \overline{Y}(\A)$ and
$\theta_2(\B) = -\overline{Y}(\B)$, so under \BD and in the absence of external
information (that is, setting $C(\Z) = 0$ for all $\Z$), the only unbiased estimator
of $\theta(\A, \B)$ is:
\begin{equation*}
\hat{\theta}(\Z) = \begin{cases}
	\overline{Y}(\A) & \mbox{ if } \Z = \A \\
	-\overline{Y}(\B) & \mbox{ if } \Z = \B \\
	0 & \mbox{otherwise}
\end{cases}
\end{equation*}
Theorem~\ref{th:bias}, however, states there exist no unbiased 
estimators for this estimand under \CRD and \CBD, both of which
assign zero probability to assignments $\A$ and $\B$.
\end{example}
Finite sample bias is not uncommon in statistical applications, and is generally acceptable
if it can be traded for a large reduction in variance and vanishes as the sample size increase. The 
next theorem studies this tradeoff for both the \BD and \CRD by looking at the 
Mean Squared Error (MSE):
\begin{theorem}\label{th:mse}
	Consider any estimand $\theta(\A,\B) = g(\VYA, \VYB)$, and suppose that $g$ is onto 
	\footnote{see appendix for technical details} [0, M].
	Then for all sample size $N$ and estimator $\hat{\theta}(\Z)$ there exist potential 
	outcomes satisfying Assumption~\ref{bounded-outcomes} such that:
	\begin{equation*}
		MSE(\hat{\theta}, \theta) \geq \frac{M^2}{8}
	\end{equation*}
	where the MSE is taken under \CRD or \BD.
\end{theorem}
In the previous section, we showed that under Assumption~\ref{no-interference} there exists an estimator
for the average total treatment effect which was unbiased, had a variance of order $O(1/N)$ and thus an MSE of
order $O(1/N)$. These properties of the estimator hold for all possible values of the potential outcomes.
Theorem~\ref{th:mse} states that under arbitrary interference, this is no longer the case: whether one is
conducting a small study or a large scale experiment on a social network, no estimator will perform well for
all values of the potential outcomes. A direct consequence of the theorem is that there exists no consistent
estimator under arbitrary interference.

\begin{remark}
	The condition on $g$ in Theorem~\ref{th:mse} can be relaxed, without affecting the main idea behind the theorem. It is useful
	to note that the difference in means estimand $\theta(\A, \B) = \overline{Y}(\A) - \overline{Y}(\B)$ 
	satisfies this condition.
\end{remark}
%

\section{Caution when structuring the interference mechanism}
\label{section:structure}

Assumptions restricting the interference mechanism (usually involving a notion of locality) can alleviate some 
of the issues mentioned in the previous section
\citep{hudgens2012toward, sinclair2012detecting, basse2016, ugander2013graph}, and allow the existence
 of unbiased estimators for a variety of estimands \citep{aronow2013estimating,ugander2013graph}. 
 Even then, however, basic intuition from the no-interference case remains misleading. 
When considering such assumptions, it is important to assess their impact on both the bias
and the variance of estimators, especially in settings which naturally afford large sample sizes, where
it is tempting to focus mostly on the bias, and ignore potential issues with the variance.
We illustrate this danger with a realistic example in which a local interference assumption
leads to an estimator that is unbiased, but whose variance explodes in large sample settings \citep{aronow2013estimating,ugander2013graph}.

If $d(i,j)$ is a measure of the distance between units $i$ and $j$, it is often plausible to assume that
only units in the \textit{k-step neighborhood} $\Ni_i^{(k)} = \{j=1\ldots N : d(i, j) \leq k\}$ of $i$  can interfere with its outcome \citep{ugander2013graph,coppock2013design}. Formally,
\begin{assumption}[k-local interference]\label{asst:local}
	\begin{equation*}
	Y_i(\Z) = Y_i(\Z_{\Ni_i^{(k)}})
	\end{equation*}	
\end{assumption}
where $\Z_{\Ni_i^{(k)}}$ denotes the sub vector of $\Z$ containing the assignments of the k-step neighbors of unit $i$. This assumption gives a special role to the following subsets of $\calZ$:
%
\begin{equation}
	\calZ_i^{(k)}(\A) = \{\Z: \Z_{\Ni_i^{(k)}} = \A_{\Ni_i^{(k)}} \} 
\end{equation}
%
where $\calZ_i^{(k)}(\B)$ is defined similarly. The set $\calZ_i^{(k)}(\A)$ is the set of assignments 
in which unit $i$ and all its $k-step$ neighbors are assigned to $A$. Consider the following Horvitz-Thompson estimator for the average causal effect:
 \begin{flalign*}
	 \hat{\theta}(\Z) &= \hat{\theta}_{\textbf{A}}(\Z) - \hat{\theta}_\textbf{B}(\Z) \\
	 &= \frac{1}{N} \sum_i \frac{\iv(Z \in \calZ_i^{(k)}(\A))}{P(\Z\in \calZ_i^{(k)}(\A)} Y_i(\A) \\
	 &- \frac{1}{N} \sum_i \frac{\iv(Z \in \calZ_i^{(k)}(\B))}{P(\Z\in \calZ_i^{(k)}(\B)} Y_i(\B)
\end{flalign*}
which can be shown to be unbiased under $k-local$ interference \citep{aronow2013estimating}. Under the 
Bernoulli Design, the variance has the following expression:
\begin{proposition}
\begin{equation*}
	\svar[\hat{\theta}] = \svar[\hat{\theta}_{\textbf{A}}] + \svar[\hat{\theta}_{\textbf{B}}] - 2 \scov[\hat{\theta}_{\textbf{A}}, \hat{\theta}_{\textbf{B}}]
\end{equation*}
where:
\begin{flalign*}
	\svar[\hat{\theta}_{\textbf{A}}] &= \frac{1}{N^2}\bigg[ \sum_i (2^{|\Ni^{(k)}_i|}-1) Y_i(\A)^2 \\
	&+ \sum_i\sum_{j\neq i}  (2^{|\Ni_i^{(k)}(\A) \cap \Ni_j^{(k)}(\A)|}-1) Y_i(\A) Y_j(\A)\bigg]
\end{flalign*}
and similarly for $\svar[\hat{\theta}_{\textbf{B}}]$, and:
\begin{flalign*}
	\scov[\hat{\theta}_{\textbf{A}}, \hat{\theta}_{\textbf{B}}] &= -\frac{1}{N^2} \bigg[ \sum_i Y_i(\A) Y_i(\B) \\
	&+ \sum_i\sum_{j\neq i} Y_i(\A) Y_j(\B) \iv(|\Ni^{(k)}_i\cap\Ni^{(k)}_j|>0)\bigg]
\end{flalign*}
\end{proposition}
The variance of the unbiased estimator thus depends explicitly on network quantities. To make things even more
explicit, the next example focuses on a specific family of networks and shows that whether the estimator is 
consistent depends on a single parameter of the network family:
%
\begin{example}\label{example:inconsistent}
For this example, we will assume that:
\begin{equation*}
	0 < K < \Y(\Z) < M
\end{equation*}
and will focus on the case where $k=1$. If we model the network as an Erdos-Renyi graph with probability $p$ of connection between nodes, we can show that:
\begin{flalign*}
\sexpg\bigg[\svar[\hat{\theta}]\bigg] &=  O\bigg(\frac{2(1+p)^{N-1} -1}{N}\\
&+ (1+3p)(1+p^2)^{(N-2)} - 1 \\
&+ \frac{1}{N} + \frac{N(N-1)}{2} (1 - (1-p)(1-p^2)^{N-2}) \bigg)
\end{flalign*}
The behavior of this quantity depends on the parameter $p$, which governs the sparsity of the network. We 
show in appendix that 
\begin{itemize}
	\item if $p < \frac{1}{N}$, we have: $\sexpg\bigg[\svar[\hat{\theta}]\bigg] \leq O(1/N)$,
	\item if $p \geq 1/\sqrt{N}$, we have: $\sexpg\bigg[\svar[\hat{\theta}]\bigg] \geq \frac{4e^{\frac{N-1}{\sqrt{N}}}}{N}K^2$, 
\end{itemize}

In this case, the expected variance of the estimator goes to zero if $p < 1/N$, but not if $p \geq 1 / \sqrt{N}$. 
\end{example}

It is straightforward to construct unbiased estimators under most forms of localized interference, by relying
on the popular Horvitz-Thompson estimators \citep{aronow2013estimating}. 
 On the other hand, checking that the variance of such estimators converges -- even under simple forms of 
 interference -- can be difficult. Yet, Example~\ref{example:inconsistent} shows the perils of neglecting
this arduous task, especially in large samples.

\section{Discussion}
\label{section:discussion}

\subsection{Broader class of estimands}

The estimands of the form $\theta(\A, \B)$, which we have considered in this article, are relevant
whenever the purpose of the experiment is to decide which of treatment $\A$ or treatment $\B$ would
be best if applied to the entire population. This is the kind of question social platforms care about when
they experiment with new products or features \citep{eckles2016estimating,gui2015network}. Other scenarios,
however, would call for different kinds of estimands. In epidemiology for instance, the question of interest is
often not which of two vaccines would be better if applied to the whole population, but which proportion of the 
population should be vaccinated \citep{hudgens2012toward}. 
It should be clear that the fundamental problems raised by arbitrary 
interference don't vanish when more complex
estimands are considered, although results analogous to Section~\ref{section:all-interferences} do become
harder to formulate as we illustrate next. For the rest of this section, consider treatment A to be an 
``active treatment'', while treatment
B is a ``control'' or ''no treatment'', and define the 
\textit{average primary causal effect}: 
\begin{equation*}
	\theta = \frac{1}{N} \sum_i Y_i(Z_i=A, \Z_{-i}=\B)
\end{equation*}
We can state the equivalent of Proposition~\ref{prop:bias-extended}:
\begin{proposition}\label{prop:bias-extended-broader}
	Denote by $\Z^{(i)}$ the assignment such that $Z_i=A$ and $Z_j = B$ for all $j\neq i$. Under 
	arbitrary interference, the only unbiased estimators of $\theta$ under the \BD are of the 
	form: 
	%
	\begin{equation}\label{eq:unbiased-primary}
	\hat{\theta}(\Z) = C(\Z) + \frac{2^N}{N} \sum_i \iv(\Z = \Z^{(i)}) Y_i(\Z^{(i)})
	\end{equation}
	where $\sum_\Z C(\Z) = 0$.
\end{proposition}
The statement of Proposition~\ref{prop:bias-extended-broader} is more complex, but the problems it highlights
are identical. Similar lines of reasoning hold for more complex estimands under arbitrary interference.

\subsection{Effective treatments and informative sets}

In Section~3, the outcome of unit $i$ depends on the assignment of units in its k-step neighborhood 
$\mathcal{N}_i^{(k)}$. This concept can be generalized by defining the \textit{reference group} of user $i$ 
\citep{manski2013identification} to be the smallest set of units $G_i \subset \{1, \ldots, N\}$ such that:
\begin{equation}\label{eq:exposure}
	Y_i(\Z) = Y(\Z_{G_i}) \quad \forall \Z
\end{equation}
This generalizes Equation~\ref{no-interference} and suggests that $Z_{G_i}$, called the \textit{effective
treatment} of unit $i$ \citep{manski2013identification,aronow2013estimating}, is more relevant than its treatment 
$Z_i$. Under Assumption~\ref{no-interference} (no interference) 
the treatment and effective treatment of unit $i$ are 
the same $\Z_{G_i} = Z_i$, while under under Assumption~\ref{asst:local}, the effective treatment of unit $i$ 
encompasses the assignment of its k-step neighbors $\Z_{G_i} = \Z_{\Ni_i^{(k)}}$. Under arbitrary 
interference the
 effective treatment of unit $i$ is the entire assignment vector, $\Z_{G_i} = \Z$. In the language of 
 Section~\ref{section:all-interferences}, an assignment vector $\Z'$ will be informative for the outcome $Y_i(\Z)$
if it results in the same effective treatment. We thus define the \textit{informative set} for $Y_i(\Z)$:
\begin{equation*}
	\calZ_i(\Z) = \{\Z'  \, : \, \Z'_{G_i} = \Z_{G_i} \quad and \quad \sprob_\R(\Z') > 0\}
\end{equation*}
Under Assumption~\ref{no-interference}, the informative set for any outcome $Y_i(\Z)$ contains every
assignments $\Z'$ such that $Z'_i = Z_i$, while under arbitrary interference, its informative set only contains
the assignment $\Z$. The next section explores how the concepts of effective treatment and informative sets 
capture the important changes implied by different interference structures. 

\subsection{Understanding interference structures}

Although we have defined effective treatments and informative sets for abstract designs $\R$, we focus the
rest of the discussion on the Bernoulli Design which captures the salient features of the problem, while 
simplifying the exposition. Denote by $E_i = |\{\Z_{G_i}\}_\Z|$ the number of effective treatments for unit
$i$ and $S_i(\Z) = |\calZ_i(\Z)|$ the size of the informative set for the outcome $Y_i(\Z)$. We call 
$F_i(\Z) = S_i(\Z) / |\calZ|$ the fraction of informative assignments for $Y_i(\Z)$. Under the Bernoulli
Design, we show in Appendix that:
\begin{equation}
	\forall \,i,\, \Z \quad F_i(\Z) = F_i  \quad \mbox{and} \quad E_i = \frac{1}{F_i} 
\end{equation}
establishing a connection between the number of effective treatments and the fraction of informative 
assignments for all units $i$. Table~\ref{table:comparisons} illustrates this relation for different interference
structures.
 \begin{table}[h!]
 \centering
 \begin{tabular}{|c|c|c|c|}
 \hline
 interference & no & 1-local & arbitrary\\
 \hline
 $E_i$ & 2 & $2^{|\Ni_i|}$ & $2^N$ \\
 $F_i(\Z)$& $\frac{1}{2}$ & $\frac{1}{2^{|\Ni_i|}}$ & $\frac{1}{2^N}$ \\
 \hline
 \end{tabular}
 \medskip
 \caption{Expected number of effective treatments and fraction of informative assignments for different 
 interference structures, under the \BD}
 \label{table:comparisons}
 \end{table}
We see that stronger assumptions on the interference structure tend to reduce the number of effective
treatments or, equivalently, increase the fraction of relevant sets. Applying these insights to 
Example~\ref{example:inconsistent}, Table~\ref{table:comparisons-2} contrasts the behaviors of  
$\sexpg[E_i]$ and $\sexpg[F_i(\Z)]$ as $N \rightarrow \infty$ for sparser networks ($p = 1/N$) and denser
networks ($p=1/\sqrt{N}$). For sparser networks, the expected number of effective treatments 
converges while the fraction of informative assignments converges to a strictly
positive number. In contrast, denser networks ($p=1/\sqrt{N}$) lead to an infinite number of effective treatments
in expectation, as $N \rightarrow \infty$.
\begin{table}[h!]
\centering
\begin{tabular}{|c|c|c|}
\hline
 & p = 1/N & $p = 1/\sqrt{N}$ \\
 \hline
$\sexpg[E_i]$ & 2e & $\infty$ \\ 
$\sexpg[F_i(\Z)]$ & $\frac{1}{2e^{1/2}}$ & 0\\
\hline
\end{tabular}
\medskip
\caption{Asymptotic expected number of effective treatments and fraction of $\Z-exposed$ sets for the two
different Erdos-Renyi specifications of Example~\ref{example:inconsistent}, under 1-local interference.}
\label{table:comparisons-2}
\end{table}
Under assumption~\ref{asst:local}, a unit's outcome depends on the 
the assignment of its neighbors and so the fewer neighbors it has, the closest it is to the no-interference 
scenario. This intuition is supported by noticing that the column of Table~\ref{table:comparisons-2}
corresponding to the sparser networks ($p=1/N$) is closer to the column of Table~\ref{table:comparisons}
corresponding to the no-interference case, while the denser networks are closer to the arbitrary interference
case. Although none of the observations we made describe sufficient conditions for the consistency of our
estimator (see \cite{aronow2013estimating}), they provide an intuitive connection between its 
asymptotic variance and the number of effective treatments (or equivalently,
the fraction of informative assignments).

\subsection{Connection with the Incidental Parameter Problem}

The focus of this paper is on design-based inference: we have considered the potential outcomes
as being fixed, and the randomness as coming from the assignment mechanism exclusively. In this section
only, we will switch viewpoints and model the potential outcomes. This alternate perspective which statisticians
tend to be more familiar with will hopefully provides additional insights into the problem.\\
Under the assumption of no interference, one might model the potential outcomes as follows:
\begin{equation*}
	(Y_i(A), Y_i(B)) \overset{iid}{\sim} \mathcal{N}\bigg( (\mu_A, \mu_B), \Sigma \bigg)
\end{equation*}
where to simplify, we assume that $\Sigma = I$. This way, there are only two parameters $\mu_A$ and $\mu_B$ to be estimated, regardless of the sample size $N$. When N increases, we have more and more units
to estimate the same number of parameters (assuming that the fraction of units in each treatment remains 
constant). This is the traditional asymptotic regime considered in statistics.

 If we allow arbitrary interference, however, the model naturally becomes:
\begin{equation*}
	\vec{Y}_i = \{Y_i(\Z)\}_{\Z\in\mathcal{Z}} \overset{iid}{\sim} \mathcal{N}\bigg( \vec{\mu}_{\mathcal{Z}}, \Sigma_{\mathcal{Z}} \bigg)
\end{equation*}
where we make the same simplification as above and assume $\Sigma = I$. The only unknowns are 
now the elements of $\vec{\mu}_{\mathcal{Z}}$, which is a vector of length $2^N$. This means that 
the number of parameters to estimate grows exponentially with the size of the problem. This scenario
is sometimes referred to as the incidental parameter problem \citep{neyman1948consistent}, and 
the challenges it poses are well known.\\

\bibliographystyle{plainnat}
\bibliography{ref}

\begin{thebibliography}{19}
\providecommand{\natexlab}[1]{#1}
\providecommand{\url}[1]{\texttt{#1}}
\expandafter\ifx\csname urlstyle\endcsname\relax
  \providecommand{\doi}[1]{doi: #1}\else
  \providecommand{\doi}{doi: \begingroup \urlstyle{rm}\Url}\fi

\bibitem[Aral and Walker(2012)]{aral2012identifying}
Sinan Aral and Dylan Walker.
\newblock Identifying influential and susceptible members of social networks.
\newblock \emph{Science}, 337\penalty0 (6092):\penalty0 337--341, 2012.

\bibitem[Aronow and Samii(2013)]{aronow2013estimating}
Peter~M Aronow and Cyrus Samii.
\newblock Estimating average causal effects under interference between units.
\newblock \emph{arXiv preprint arXiv:1305.6156}, 2013.

\bibitem[Basse and Feller(2016)]{basse2016}
G.~Basse and A.~Feller.
\newblock Analyzing multilevel experiments in the presence of peer effects.
\newblock \emph{ArXiv e-prints}, August 2016.

\bibitem[Basse and Airoldi(2015)]{basse2015optimal}
Guillaume~W Basse and Edoardo~M Airoldi.
\newblock Optimal design of experiments in the presence of network-correlated
  outcomes.
\newblock \emph{ArXiv e-prints}, 2015.

\bibitem[Basu(2011)]{basu2011essay}
Debabrata Basu.
\newblock An essay on the logical foundations of survey sampling, part one.
\newblock In \emph{Selected Works of Debabrata Basu}, pages 167--206. Springer,
  2011.

\bibitem[Bowers et~al.(2016{\natexlab{a}})Bowers, Desmarais, Frederickson,
  Ichino, Lee, and Wang]{bowers2016ab}
J.~Bowers, B.~A. Desmarais, M.~Frederickson, N.~Ichino, {H.-W.} Lee, and
  S.~Wang.
\newblock Models, methods and network topology: {E}xperimental design for the
  study of interference.
\newblock arxiv paper no. 1601.00992, January 2016{\natexlab{a}}.

\bibitem[Bowers et~al.(2016{\natexlab{b}})Bowers, Fredrickson, and
  Aronow]{bowers2016aa}
J.~Bowers, M.~M. Fredrickson, and P.~M. Aronow.
\newblock A more powerful test statistic for reasoning about interference
  between units.
\newblock \emph{Political Analysis}, 24\penalty0 (3):\penalty0 395--403,
  2016{\natexlab{b}}.

\bibitem[Coppock and Sircar(2013)]{coppock2013design}
Alexander Coppock and Neelanjan Sircar.
\newblock `design of field experiments under unknown interference structures.
\newblock \emph{Work. Pap}, 2013.

\bibitem[Eckles et~al.(2016)Eckles, Kizilcec, and Bakshy]{eckles2016estimating}
Dean Eckles, Ren{\'e}~F Kizilcec, and Eytan Bakshy.
\newblock Estimating peer effects in networks with peer encouragement designs.
\newblock \emph{Proceedings of the National Academy of Sciences}, 113\penalty0
  (27):\penalty0 7316--7322, 2016.

\bibitem[Gui et~al.(2015)Gui, Xu, Bhasin, and Han]{gui2015network}
Huan Gui, Ya~Xu, Anmol Bhasin, and Jiawei Han.
\newblock Network a/b testing: From sampling to estimation.
\newblock In \emph{Proceedings of the 24th International Conference on World
  Wide Web}, pages 399--409. ACM, 2015.

\bibitem[Hudgens and Halloran(2012)]{hudgens2012toward}
Michael~G Hudgens and M~Elizabeth Halloran.
\newblock Toward causal inference with interference.
\newblock \emph{Journal of the American Statistical Association}, 2012.

\bibitem[Imbens and Rubin(2015)]{imbens2015causal}
Guido~W Imbens and Donald~B Rubin.
\newblock \emph{Causal inference in statistics, social, and biomedical
  sciences}.
\newblock Cambridge University Press, 2015.

\bibitem[Manski(2013)]{manski2013identification}
Charles~F Manski.
\newblock Identification of treatment response with social interactions.
\newblock \emph{The Econometrics Journal}, 16\penalty0 (1):\penalty0 S1--S23,
  2013.

\bibitem[Neyman and Scott(1948)]{neyman1948consistent}
Jerzy Neyman and Elizabeth~L Scott.
\newblock Consistent estimates based on partially consistent observations.
\newblock \emph{Econometrica: Journal of the Econometric Society}, pages 1--32,
  1948.

\bibitem[Rubin(1980)]{rubin1980comment}
Donald~B Rubin.
\newblock Comment.
\newblock \emph{Journal of the American Statistical Association}, 75\penalty0
  (371):\penalty0 591--593, 1980.

\bibitem[S{\"a}rndal et~al.(2003)S{\"a}rndal, Swensson, and
  Wretman]{sarndal2003model}
Carl-Erik S{\"a}rndal, Bengt Swensson, and Jan Wretman.
\newblock \emph{Model assisted survey sampling}.
\newblock Springer Science \& Business Media, 2003.

\bibitem[Shalizi and Thomas(2011)]{shalizi2011homophily}
Cosma~Rohilla Shalizi and Andrew~C Thomas.
\newblock Homophily and contagion are generically confounded in observational
  social network studies.
\newblock \emph{Sociological Methods \& Research}, 40\penalty0 (2):\penalty0
  211--239, 2011.

\bibitem[Sinclair et~al.(2012)Sinclair, McConnell, and
  Green]{sinclair2012detecting}
Betsy Sinclair, Margaret McConnell, and Donald~P Green.
\newblock Detecting spillover effects: Design and analysis of multilevel
  experiments.
\newblock \emph{American Journal of Political Science}, 56\penalty0
  (4):\penalty0 1055--1069, 2012.

\bibitem[Ugander et~al.(2013)Ugander, Karrer, Backstrom, and
  Kleinberg]{ugander2013graph}
Johan Ugander, Brian Karrer, Lars Backstrom, and Jon Kleinberg.
\newblock Graph cluster randomization: Network exposure to multiple universes.
\newblock In \emph{Proceedings of the 19th ACM SIGKDD international conference
  on Knowledge discovery and data mining}, pages 329--337. ACM, 2013.

\end{thebibliography}

\appendix
\section{Technical appendix}

\subsection{Variance terms in Section~2}

\begin{flalign*}
	V_A &= \frac{1}{N-1} \sum_{i=1}^N\bigg(Y_i(A) - \overline{Y}(A)\bigg)^2 \\
	V_B &= \frac{1}{N-1} \sum_{i=1}^N\bigg(Y_i(B) - \overline{Y}(B)\bigg)^2 \\
	V_{\theta} &= \frac{1}{N-1} \sum_{i=1}^N \bigg( Y_i(A) - Y_i(B) - \theta \bigg)^2
\end{flalign*}

\subsection{Proof of theorems}

\begin{theorem}\label{th:bias}
	Consider any estimand $\theta(\A, \B) = g(\VYA, \VYB)$ that is not constant, and any 
	assignment mechanism $\R$ such that $\sprob_\R(\Z = \A) = \sprob_\R(\Z=\B) = 0$. 
	There exists no unbiased estimator of $\theta(\A,\B)$ under $\R$. 
\end{theorem}

\begin{proof}
The key insight behind the proof is that under arbitrary interference, if $\Z$ and $\Z'$ are 
two assignments such that $\Z \neq \Z'$, then $Y(\Z)$ and $Y(\Z')$ can be seen as two 
independent parameters. That is, the value of $Y(\Z)$ does not constrain that of $Y(\Z')$.
Denote by $p(\Z_0)$ the probability $\sprob_\R(\Z = \Z_0)$, and assume that there exists an
unbiased estimator $\hat{g}(Y(\Z))$. This means that:
\begin{equation*}
	\sum_{\Z \in \calZ - \{\A,\B\}}p(\Z)\hat{g}(Y(\Z)) = g(Y(\A), Y(\B))
\end{equation*}
since $p(\A) = p(\B) = 0$. Now the LHS does not depend on the independent 
parameters $Y(\A)$ or $Y(\B)$, and so the 
RHS must not depend on $Y(\A)$ or $Y(\B)$ either. But $RHS$ is only a function of these 
two parameters, so we must have: $g(Y(\A), Y(\B)) = C$, where $C$ is a constant. This violates
the hypothesis that $g$ is not degenerate. Both CRD and CRB satisfy $p(\A) = p(\B) = 0$, so the
result holds for both designs.
\end{proof}

\begin{proposition}\label{prop:bias-extended}
	Consider the Bernoulli Design. If the estimand is of additive form 
	$\theta(\A,\B) = \theta_1(\A) + \theta_2(\B)$, then unbiased estimators are of the form 
	\begin{equation*}
		\hat{\theta}(\Z) = C(\Z) + 2^N \iv(\Z = \A)\theta_1(\A) + 2^N\iv(\Z = \B)\theta_2(\B)
	\end{equation*}
	where $C(\Z)$ does not depend on any potential outcomes and satisfies:
	\begin{equation*}
		\sum_{\Z\in\mathcal{Z}} C(\Z) = 0
	\end{equation*} 
	For other types of estimands $\theta(\A,\B)$, there exist no unbiased estimators.
\end{proposition}
\begin{proof}
Recall that we use the notation $\theta_1(\A) = g_1(\VYA)$ and $\theta_2(\B) = g_2(\VYB)$. Let $\R$ be the Bernoulli Design. Suppose that $\hat{\theta}(\Z) = \hat{g}(\VYZ)$ is unbiased 
for $g(\VYA, \VYB)$. 
We have:
\begin{flalign}
E(\hat{g}(\VYZ)) = g(\VYA, \VYB) \quad &\Leftrightarrow \quad \frac{1}{|\mathcal{Z}|} \sum_{Z \in \mathcal{Z}} \hat{g}(\VYZ) = g(\VYA, \VYB) \notag\\
&\Leftrightarrow \frac{1}{|\mathcal{Z}|} \sum_{Z \in \mathcal{Z} - \{A,B\}} \hat{g}(\VYZ) + \frac{1}{|\mathcal{Z}|}\hat{g}(\VYA) + \frac{1}{|\mathcal{Z}|}\hat{g}(\VYB)\\
&= g(\VYA,\VYB) \notag\\
&\Leftrightarrow g(\VYA,\VYB) - \frac{1}{|\mathcal{Z}|}\hat{g}(\VYA) - \frac{1}{|\mathcal{Z}|}\hat{g}(\VYB)\\
&= \frac{1}{|\mathcal{Z}|} \sum_{Z \in \mathcal{Z} - \{A,B\}} \hat{g}(\VYZ) \label{eq:init} 
\end{flalign}
Now RHS of the equation doesn't depend on $\VYA$ or $\VYB$ so the LHS mustn't either. Moreover, the 
LHS only depends on $\VYA$ and $\VYB$, so in the end, the LHS cannot depend on any potential outcome $\{Y(Z)\}_Z$. That is, we must have:
\begin{equation} \label{eq:cst}
	g(\VYA, \VYB) - \frac{1}{|\mathcal{Z}|}\hat{g}(\VYA) - \frac{1}{|\mathcal{Z}|}\hat{g}(\VYB) = C^*
\end{equation}
and so:
\begin{equation*}
	|\mathcal{Z}| g(\VYA, \VYB) -\hat{g}(\VYA) - \hat{g}(\VYB) = C
\end{equation*}
where $C$ does not depend on any potential outcome. But then this means that
\begin{equation*}
	|\mathcal{Z}| g(\VYA, \VYB) = C +  \hat{g}(\VYB) + \hat{g}(\VYA) 
\end{equation*}
But $\hat{g}(\VYB)$ and $\hat{g}(\VYA)$ are functions of $\VYB$ and $\VYA$ respectively, so 
this means that if the estimand $g$ can be estimated unbiasedly, it has to be of the form:
\begin{flalign*}
	g(\VYA, \VYB) &= C^* +  g_1(\VYA)  + g_2(\VYB) = C^* + \theta^*_1(\A) + \theta^*_2(\B)\\
	&= \theta_1(\A) + \theta_2(\B)
\end{flalign*}
since $C^*$ doesn't depend on any potential outcome (for instance, let $\theta_1(\A) = C^* + g_1(\VYA)$ and
$\theta_2(\B) = g_2(\VYB)$...).
Now consider any estimand of the form $g(\VYA, \VYB) = C^* +  g_1(\VYA)  + g_2(\VYB)$. Plugging
into Equation~\ref{eq:cst}, we have:
\begin{equation*}
	g_1(\VYA) - \hat{g}(\VYA) + g_2(\VYB) - \hat{g}(\VYB) = C
\end{equation*}
and  so:
\begin{equation*}
	g_1(\VYA) -  \frac{1}{|\mathcal{Z}|} \hat{g}(\VYA)  = C -  g_2(\VYB) +   \frac{1}{|\mathcal{Z}|}\hat{g}(\VYB)
\end{equation*}
now we use the same reasoning as above, leading to the fact that:
\begin{equation*}
 \hat{g}(\VYA) = C_1 + |\mathcal{Z}|g_1(\VYA)
\end{equation*}
and 
\begin{equation*}
 \hat{g}(\VYB) = C_2 +  |\mathcal{Z}|g_2(\VYB)
\end{equation*}
Return now to Equation~\ref{eq:init}, focusing on the RHS and let $\Z_0 \not \in \{\A,\B\}$. A similar argument
as the one used above gives:
\begin{flalign*}
	\frac{1}{|\mathcal{Z}|} \sum_{\Z \in \mathcal{Z} - \{\A,\B\}} \hat{g}(\VYZ)  = C^*\Leftrightarrow 
	\frac{1}{|\mathcal{Z}|} \sum_{\Z \in \mathcal{Z} - \{\A,\B,\Z_0\}} \hat{g}(\VYZ) = C^* - \frac{1}{|\mathcal{Z}|} \hat{g}(\VYZo)
\end{flalign*}
using the same reasoning once again, we conclude that:
\begin{equation*}
	\hat{g}(\VYZo) = C
\end{equation*}
where C does not depend on any potential outcome. And so putting it all together, we have:
\begin{equation*}
	\hat{g}(\VYZ) = C(Z) + \frac{\iv(Z = A)}{1 / |\mathcal{Z}|}g_1(\VYA) + \frac{\iv(Z = B)}{1 / |\mathcal{Z}|}g_1(\VYB)
\end{equation*}

\end{proof}

\begin{theorem}
	Consider any estimand $\theta(\A,\B) = g(\VYA, \VYB)$, and suppose that $g$ is onto the interval [0, M]. 
	That is,
	\begin{equation*}
	\forall x \in [0,M], \, \exists \,\,\, \VYA, \VYB\, \quad s.t\quad g(\VYA, \VYB) = x
	\end{equation*}
	Then for all sample size $N$ and estimator $\hat{\theta}(\Z)$ there exist potential 
	outcomes satisfying Assumption~1 such that:
	\begin{equation*}
		MSE(\hat{\theta}, \theta) \geq \frac{M^2}{8}
	\end{equation*}
	where the MSE is taken under \CRD or \BD.
\end{theorem}

\begin{proof}
We write down the MSE:
\begin{flalign*}
	MSE(\hat{g}(\VYZ), g(\VYA,\VYB)) &= E\bigg[(\hat{g}(\VYZ) - g(\VYA,\VYB))^2\bigg]
\end{flalign*}
Consider potential outcomes $\{\VYZ\}_Z$ such that $\VYZ = \Y = cst$ for all $\Z\not\in \{\A,\B\}$, and let 
$C = \hat{g}(\Y)$. We now need to consider separately the case of Bernoulli Design (BD) 
and Completely Randomized Design (CRD).  We
start with the CRD. Let $\mathcal{Z}_{N_A}$ be the set of all assignments with $N_A$ units assigned to treatment $A$. We have:
\begin{flalign*}
MSE(\hat{g}(\VYZ), g(\VYA,\VYB)) &= E\bigg[(\hat{g}(\VYZ) - g(\VYA,\VYB))^2\bigg]\\
&= \frac{1}{|\mathcal{Z}_{N_A}|} \sum_{Z\in \mathcal{Z}_{N_A}} (\hat{g}(\VYZ) - g(\VYA,\VYB))^2\\
&= \frac{1}{|\mathcal{Z}_{N_A}|} \sum_{Z\in \mathcal{Z}_{N_A}} (C - g(\VYA,\VYB))^2 \\
&= (C - g(\VYA,\VYB))^2
\end{flalign*}
Since g is surjective, chose $\VYA$ and $\VYB$ such that:
\begin{equation*}
	g(\VYA, \VYB) = \begin{cases}
		0 &\mbox{if} \quad |C - M| < M / 2 \\
		M &\mbox{otherwise}
	\end{cases}
\end{equation*}
and so:
\begin{equation*}
MSE(\hat{g}(\VYZ), g(\VYA, \VYB)) = (C - g(\VYA, \VYB))^2 > \frac{M^2}{4} > \frac{M^2}{8}
\end{equation*}
\end{proof}
now let's turn to the BD case. The slight difference is that the assignments $\Z=\A$ and $\Z=\B$ may occur. We
thus have:
\begin{flalign*}
MSE(\hat{g}(\VYZ), g(\VYA,\VYB)) &= E\bigg[(\hat{g}(\VYZ) - g(\VYA,\VYB))^2\bigg] \\
&= \frac{1}{|\mathcal{Z}|} \sum_{\Z\in \mathcal{Z}} (\hat{g}(\VYZ) - g(\VYA,\VYB))^2\\
&= \frac{1}{|\mathcal{Z}|} \sum_{\Z\in \mathcal{Z} - \{\A,\B\}} (\hat{g}(\VYZ) - g(\VYA,\VYB))^2 \\
&+  \frac{1}{|\mathcal{Z}|} (\hat{g}(\VYA) - g(\VYA,\VYB))^2 \\
&+ \frac{1}{|\mathcal{Z}|} (\hat{g}(\VYB) - g(\VYA,\VYB))^2\\
&= \frac{1}{|\mathcal{Z}|} \sum_{Z\in \mathcal{Z} - \{\A,\B\}} (C - g(\VYA,\VYB))^2 \\ 
&+  \frac{1}{|\mathcal{Z}|} (\hat{g}(\VYA) - g(\VYA,\VYB))^2  \\
&+ \frac{1}{|\mathcal{Z}|} (\hat{g}(\VYB) - g(\VYA,\VYB))^2\\
&\geq \frac{|\mathcal{Z}| - 2}{|\mathcal{Z}|} (C - g(\VYA,\VYB))^2  \\
&\geq \frac{|\mathcal{Z}| - 2}{|\mathcal{Z}|} \frac{M^2}{4}\\
&\geq \frac{M^2}{8}
\end{flalign*}
which concludes the proof.

\begin{remark}
The surjective assumption is not crucial to the spirit of the result, but it does make it easier to state the 
theorem, and highlight the magnitude of how much things can go wrong. Note that the popular average total
effect estimand satisfies the conditions of the theorem.
\end{remark}


\begin{proposition}
\begin{equation*}
	\svar[\hat{\theta}] = \svar[\hat{\theta}_{\textbf{A}}] + \svar[\hat{\theta}_{\textbf{B}}] - 2 \scov[\hat{\theta}_{\textbf{A}}, \hat{\theta}_{\textbf{B}}]
\end{equation*}
where:
\begin{flalign*}
	\svar[\hat{\theta}_{\textbf{A}}] &= \frac{1}{N^2}\bigg[ \sum_i (2^{|\Ni^{(k)}_i|}-1) Y_i(\A)^2 \\
	&+ \sum_i\sum_{j\neq i}  (2^{|\Ni_i^{(k)}(\A) \cap \Ni_j^{(k)}(\A)|}-1) Y_i(\A) Y_j(\A)\bigg]
\end{flalign*}
and similarly for $\svar[\hat{\theta}_{\textbf{B}}]$, and:
\begin{flalign*}
	\scov[\hat{\theta}_{\textbf{A}}, \hat{\theta}_{\textbf{B}}] &= -\frac{1}{N^2} \bigg[ \sum_i Y_i(\A) Y_i(\B) \\
	&+ \sum_i\sum_{j\neq i} Y_i(\A) Y_j(\B) \iv(|\Ni^{(k)}_i\cap\Ni^{(k)}_j|>0)\bigg]
\end{flalign*}
\end{proposition}

\begin{proof}
Throughout, we will let $I_i(\A) = \iv(\Z \in \calZ_i^{(k)}(\A))$ and $P_i(\A) = P(I_i(\A) = 1)$.
We begin by looking at $\svar(\hat{\theta}_A)$. We have:
\begin{equation*}
	\svar(\hat{\theta}_A) = \frac{1}{N^2}\bigg[\sum_i \bigg(\frac{Y_i(\A)}{P_i(\A)}\bigg)^2 \svar(I_i(\A)) 
	+ \sum_i\sum_{j\neq i} \frac{Y_i(\A)Y_j(\A)}{P_i(\A)P_j(\A)} \scov(I_i(\A), I_j(\A))\bigg]
\end{equation*}
But we also have:
\begin{equation*}
	\svar(I_i(\A)) = E(I_i(\A)^2) - E(I_i(\A))^2 = P_i(\A) - P_i(\A)^2 = P_i(\A)(1-P_i(\A))
\end{equation*}
and:
\begin{equation*}
	\scov(I_i(\A), I_j(\A)) = E(I_i(\A) I_j(\A)) - E(I_i(\A)) E(I_j(\A)) = P_{ij}(\A,\A) - P_i(\A)P_j(\A)
\end{equation*}
where $P_{ij}(\A, \A) = P\bigg(I_i(\A) = 1, I_j(\A)=1\bigg)$. And so finally:
\begin{flalign*}
	\svar(\hat{\theta}_A) &= \frac{1}{N^2}\bigg[\sum_i \bigg(\frac{Y_i(\A)}{P_i(\A)}\bigg)^2 P_i(\A)(1-P_i(\A))
	+ \sum_i\sum_{j\neq i} \frac{Y_i(\A)Y_j(\A)}{P_i(\A)P_j(\A)}(P_{ij}(\A,\A) - P_i(\A)P_j(\A))\bigg] \\
	&= \frac{1}{N^2}\bigg[\sum_i (\frac{1}{P_i(\A)} - 1)Y_i(\A)^2 
	+ \sum_i\sum_{j\neq i} (\frac{P_{ij}(\A,\A)}{P_i(\A)P_j(\A)} - 1)Y_i(\A)Y_j(\A)\bigg] 
\end{flalign*}
But under k-local exposure and Bernoulli Design, we have $P_i(\A) = \bigg(\frac{1}{2}\bigg)^{|\Ni_i^{(k)}|}$ and
\begin{flalign*}
	P(I_i(\A), I_j(\A)) &= P(I_j(\A) | I_i(\A)) P_i(\A) \\
	&= \bigg( \frac{1}{2} \bigg)^{|\Ni_j^{(k)}| - |\Ni^{(k)}_i\cap \Ni^{(k)}_j|} \bigg(\frac{1}{2}\bigg)^{|\Ni^{(k)}_i|}
\end{flalign*}
and so, plugging into the equation for $\svar(\hat{\theta}_A)$, we have.
\begin{equation*}
	\svar(\hat{\theta}_A) = \frac{1}{N^2} \bigg[ \sum_i \bigg(2^{|\Ni^{(k)}_i} - 1\bigg)Y_i(\A)^2 +
	\sum_i \sum_{j\neq i} \bigg(2^{|\Ni_i^{(k)|} \cap \Ni_j^{(k)}|} - 1\bigg) Y_i(\A) Y_j(\A) \bigg]
\end{equation*}
We now turn to the covariance term $\scov(\hat{\theta}_A, \hat{\theta}_B)$. We have:
\begin{equation*}
	\scov(\hat{\theta}_A, \hat{\theta}_B) = \frac{1}{N^2}\bigg[\sum_i \frac{Y_i(\A)Y_i(\B)}{P_i(\A)P_i(\B)} \scov(I_i(\A), I_i(\B)) + \sum_i\sum_{j\neq i} \frac{Y_i(\A)Y_j(\B)}{P_i(\A)P_j(\B)} \scov(I_i(\A), I_j(\B))  \bigg]
\end{equation*}
but we have:
\begin{flalign*}
\scov(I_i(\A), I_i(\B)) &= E(I_i(\A) I_i(\B)) - E(I_i(\A)) E(I_i(\B)) \\
&= P\bigg(I_i(\A) = 1, I_i(\B)=1\bigg) - P_i(\A) P_i(\B) \\
&= -P_i(\A) P_i(\B)
\end{flalign*}
and
\begin{flalign*}
\scov(I_i(\A), I_j(\B)) &= E(I_i(\A) I_j(\B)) - E(I_i(\A)) E(I_j(\B))\\
&= P(I_j(\B) = 1 | I_i(\A) = 1) P_i(\A) - P_i(\A)P_j(\B) \\
&= \iv(|\Ni^{(k)}_i \cap \Ni^{(k)}_j| = 0) P_j(\B) P_i(\A) - P_i(\A)P_j(\B) \\
\end{flalign*}
and so putting it all together, we have:
\begin{equation*}
	\scov(\hat{\theta}_A, \hat{\theta}_B) = -\frac{1}{N^2}\bigg[\sum_i Y_i(\A)Y_i(\B) + \sum_i\sum_{j\neq i} \bigg(1-\iv(|\Ni^{(k)}_i \cap \Ni^{(k)}_j| = 0) \bigg)Y_i(\A)Y_j(\B)  \bigg]
\end{equation*}
which completes the proof, since $1-\iv(|\Ni^{(k)}_i \cap \Ni^{(k)}_j| = 0) = \iv(|\Ni^{(k)}_i \cap \Ni^{(k)}_j| > 0) $.

\end{proof}

\begin{proposition}
	Denote by $\Z^{(i)}$ the assignment such that $Z_i=A$ and $Z_j = B$ for all $j\neq i$. Under 
	arbitrary interference, the only unbiased estimators of $\theta =  \frac{1}{N} \sum_i Y_i(Z_i=A, \Z_{-i}=\B)$ under the \BD are of the 
	form: 
	%
	\begin{equation}\label{eq:unbiased-primary}
	\hat{\theta}(\Z) = C(\Z) + 2^N \sum_i \iv(\Z = \Z^{(i)}) Y_i(\Z^{(i)})
	\end{equation}
	where $\sum_\Z C(\Z) = 0$.
\end{proposition}

\begin{proof}
The proof follows the same lines as that of Proposition~1. Let $\hat{\theta}(\Z) = \hat{g}(\VYZ)$ be an 
unbiased estimator of $\theta$ under the $\BD$. Thus we have:
\begin{flalign}
	\frac{1}{|\mathcal{Z}|} \sum_{\Z \in \mathcal{Z}} \hat{g}(\VYZ) 
	&= \frac{1}{N}\sum_i Y_i(Z_i=A, \Z_{-i}=\B) \notag\\ 
	\Rightarrow \quad \frac{1}{|\mathcal{Z}|} \sum_{\Z \in \mathcal{Z} - \cup_i \Z^{(i)}} \hat{g}(\VYZ) 
	&= \frac{1}{N}\sum_i Y_i(Z_i=A, \Z_{-i}=\B) - \frac{1}{|\mathcal{Z}|} \sum_{\cup_i \Z^{(i)}} \hat{g}(\VYZ)\notag\\
\end{flalign}
The LHS does not depend on any $\Y(\Z^{(i)})$ for $i=1, \ldots, N$ so the RHS mustn't either. But the RHS only
depends on $\Y(\Z^{(i)})$ for $i=1, \ldots, N$, so it must be that $RHS$ doesn't depend on any potential
outcome. That is, 
\begin{equation*}
 \frac{1}{N}\sum_i Y_i(Z_i=A, \Z_{-i}=\B) - \frac{1}{|\mathcal{Z}|} \sum_{\cup_i \Z^{(i)}} \hat{g}(\VYZ) = C
\end{equation*}
where $C(\Z)$ does not depend on any potential outcomes. Now consider an index $i_0$. We have, using
the previous equation:
\begin{equation*}
	\frac{1}{N}\sum_{i\neq i_0} Y_i(Z_i=A, \Z_{-i}=\B) - \frac{1}{\mathcal{Z}} \sum_{\cup_i\Z^{(i)} - \Z^{(i_0)}} \hat{g}(\Y(\Z)) = \frac{1}{\mathcal{Z}} \hat{g}(\Y(\Z^{(i_0)})) -  \frac{1}{N} \Y_{i_0}(\Z^{(i_0)})
\end{equation*}
Again, the LHS does not depend on $\Y(\Z^{(i_0)})$ so the RHS shouldn't either. But the RHS depends 
explicitly on $\Y(\Z^{(i_0)})$ , so it must be that the RHS doesn't depend on any potential outcomes. Hence:
\begin{flalign*}
C(\Z^{(i_0)}) + \frac{1}{|\mathcal{Z}|} \hat{g}(\Y(\Z^{(i_0)})) -  \frac{1}{N} \Y_{i_0}(\Z^{(i_0)}) &= C'(\Z^{(i_0)})\\
\Rightarrow \hat{g}(\Y(\Z^{(i_0)}))  = C^*(\Z^{(i_0)}) + \frac{|\mathcal{Z}|}{N} \Y_{i_0}(\Z^{(i_0)})
\end{flalign*}
That is, renaming $C^*(\Z) \rightarrow  C(\Z)$, we have:
\begin{equation*}
	\hat{\theta}(\Z^{(i_0)}) = C(\Z^{(i_0)}) + \frac{|\mathcal{Z}|}{N} Y_{i_0}(\Z^{(i_0)})
\end{equation*}	
Now we could apply the same reasoning to any index $i$, so:
\begin{equation*}
\hat{\theta}(\Z^{(i)}) = C(\Z^{(i)}) + \frac{|\mathcal{Z}|}{N} Y_{i}(\Z^{(i)})
\end{equation*}
Now going back to the first equation of the proof, the RHS does not depend on any potential 
outcomes either. So for any $\Z_0 \in \mathcal{Z} - \cup_i \Z^{(i)}$:
\begin{equation*}
	\frac{1}{|\mathcal{Z}|} \sum_{\Z \in \mathcal{Z} - \cup_i \Z^{(i)} - \Z_0} \hat{g}(\VYZ) 
	 = C  - \frac{1}{|\mathcal{Z}|} \hat{g}(\Y(\Z_0))  \\
\end{equation*}
here again, the $LHS$ does not depend on $\Y(\Z_0)$ so the $RHS$ shouldn't either. But the RHS
depends explicitly only on $\Y(\Z_0)$, and so it must be that it depends on no potential outcomes. Hence:
\begin{equation*}
	\hat{g}(\Z_0) = C(\Z_0)
\end{equation*}
Finally putting it all together we have:
\begin{equation*}
\hat{\theta}(\Z) = C(\Z) + \begin{cases}
  \frac{|\mathcal{Z}|}{N} Y_{i}(\Z) & \mbox{ if } \Z = \Z^{(i)}\\
  0  & \mbox{ otherwise}
\end{cases}
\end{equation*}
this can be rewritten:
\begin{equation*}
\hat{\theta}(\Z) = C(\Z) + \frac{2^N}{N} \sum_{\mathcal{Z}} \iv(\Z = \Z^{(i)}) Y_i(\Z^{(i)})
\end{equation*}

\end{proof}

\subsection{Derivations for Example~3}

We now assume that $k=1$, and the networks are $ER(N, p)$. We also assume that there 
exists $K, M > 0$ such that:
\begin{equation*}
	K < Y_i(\Z) < M \quad \forall i = 1 \ldots N
\end{equation*}

Define:
\begin{flalign*}
	h_N(C) &= \frac{2}{N^2}\bigg[\sum_i E_\G\bigg(2^{|\Ni^{(k)}_i|} - 1\bigg)C^2  \\
	&+ \sum_i\sum_{j\neq i}  E_\G\bigg(2^{|\Ni_i^{(k)|} \cap \Ni_j^{(k)}|} - 1\bigg)C^2 \\
	&+ N C^2  \\
	&+ \sum_i\sum_{j\neq i} E_\G\bigg(1-\iv(|\Ni^{(k)}_i \cap \Ni^{(k)}_j| = 0)  \bigg)C^2\bigg]
\end{flalign*}
Since all the quantities are positive, it is easy to verify that:
\begin{equation*}
	h_N(K) \leq E_\G(\svar(\hat{\theta})) \leq h_N(M)
\end{equation*}
Under the assumptions stated for this section, we have:
\begin{equation*}
	E_\G\bigg(2^{|\Ni^{(k)}_i|} - 1\bigg) = 2E_\G(2^{d_i}) - 1 \\
\end{equation*}
where $d_i$ is the degree of node $i$. Under the $ER(N,p)$ model, we have $d_i\sim Bin(N-1,p)$. It is then
easy to verify (via the moment generating function) that:
\begin{equation*}
	E_\G(2^{d_i}) = (2p + (1-p))^{N-1} = (1+p)^{(N-1)}
\end{equation*}
and so:
\begin{equation*}
	E_G\bigg(2^{|\Ni^{(k)}_i} - 1\bigg) = 2(1+p)^{(N-1)} - 1
\end{equation*}
Similarly, we have:
\begin{equation*}
	|\Ni^{(k)}_i \cap \Ni^{(k)}_j| = 2I_{ij} + C_{ij}
\end{equation*}	
where $I_{ij} \sim Bern(p)$, $C_{ij} \sim Bin(N-2, p^2)$, and $I_{ij} \perp C_{ij}$. So we have:
\begin{flalign*}
	E_\G\bigg(2^{|\Ni^{(k)}_i \cap \Ni^{(k)}_j|}\bigg) &= E(2^{2I_{ij}})E(2^{C_{ij}})\\
	&= (3p+1) (2p^2 + (1-p^2))^{N-2} \\
	&= (3p+1)(p^2 + 1)^{N-2}
\end{flalign*}
and:
\begin{equation*}
	E_\G\bigg(1-\iv(|\Ni_i \cap \Ni_j| = 0)\bigg) = 1 - P_\G(|\Ni_i \cap \Ni_j| = 0) = 1 - (1-p)(1-p^2)^{N-2}
\end{equation*}
and so
\begin{flalign*}
	h_N(C) &= 2 \bigg[  \frac{2(1+p)^{N-1} - 1}{N} C^2 \\
	&+  \bigg( (3p+1)(p^2 + 1)^{N-2} -1\bigg) C^2 \\
	&+ \frac{C^2}{N} \\
	&+  \bigg(1 - (1-p)(1-p^2)^{N-2} \bigg)C^2\bigg]
\end{flalign*}
In fact, we will augment this notation a bit, and denote it by $h_N(C,p)$, to mark 
the dependence on $C$ and on $p$. Note that we can verify that $h_N(C,p)$ is 
monotone non-decreasing in $p$.
\paragraph{Sparse graphs: $p< \frac{1}{N}$:}
Let $p_N = \frac{1}{N}$, we have:
\begin{flalign*}
	2(1+p_N)^{N-1} - 1 &= 2e^{(N-1)log(1+\frac{1}{N})} - 1 \\
	&\underset{N\rightarrow \infty}{=} 2e^{(N-1) / N} - 1\\
	&\underset{N\rightarrow \infty}{=}2e\times e^{(N-1) / N - 1} - 1 \\
	&\underset{N\rightarrow \infty}{=} 2e( 1 + \frac{(N-1)}{N} - 1) - 1\\
	&= O(1)
\end{flalign*}
and so:
\begin{equation*}
\frac{2(1+p_N)^{N-1} - 1 }{N} = O\bigg(\frac{1}{N}\bigg)
\end{equation*}
Similarly, we have:
\begin{flalign*}
	(3p_N + 1)(p_N^2 + 1)^{N-2} - 1&= (\frac{3}{N} + 1)e^{(N-2)log(1+\frac{1}{N^2})} =1\\
	&\underset{N\rightarrow \infty}{=} (\frac{3}{N} + 1)e^{\frac{(N-2)}{N^2}}-1 \\
	&\underset{N\rightarrow \infty}{=} (\frac{3}{N} + 1)(1 + O\bigg(\frac{N-2}{N^2}\bigg)) -1\\
	&= O\bigg(\frac{1}{N}\bigg) 
\end{flalign*}
and finally, we have:
\begin{flalign*}
	1 - (1-p_N)(1-p_N^2)^{N-2}  &= 1 - (1-p_N)e^{(N-2)log(1-p_N^2)}\\
	&\underset{N\rightarrow\infty}{=} 1 - (1-\frac{1}{N})e^{-\frac{N-2}{N^2}}\\
	&\underset{N\rightarrow\infty}{=}  1 - (1-\frac{1}{N})(1 - \frac{N-2}{N^2}) \\
	&\underset{N\rightarrow\infty}{=}  = O\bigg(\frac{1}{N}\bigg)
\end{flalign*}
And so putting it all together, we have $h_N(M,p_N) = O(\frac{1}{N})$ and so finally:
\begin{equation*}
	E_\G(\svar(\hat{\theta})) \leq h_N(M,p) \leq h_N(M, p_N) = O\bigg(\frac{1}{N}\bigg)
\end{equation*}

\paragraph{Dense graphs: $p> \frac{1}{\sqrt{N}}$:}
Let $p_N = \frac{1}{\sqrt{N}}$. We have:
\begin{flalign*}
	\frac{2(1+p_N)^{N-1} - 1}{N} &= \frac{2e^{(N-1) log(1 + \frac{1}{\sqrt{N}})}}{N}\\
	&\underset{N\rightarrow\infty}{=} \frac{2e^{\frac{(N-1)}{\sqrt{N}}}}{N}
\end{flalign*}
and so:
\begin{equation*}
	E_\G(\svar(\hat{\theta}))  \geq h_N(K, p) \geq h_N(K,p_N) \geq  \frac{4e^{\frac{(N-1)}{\sqrt{N}}}}{N} K^2
\end{equation*}

\subsection{Derivations for Section~4}

\paragraph{$E_i$ , $S_i(\Z)$  and $F_i$ under BD:}

We have seen that under no-interference, $G_i = \{i\}$ so the effective treatments are $\{\Z_{\{G_i\}}\}_{\Z\in\calZ} = \{A,B\}$
and so $E_i = 2$. Under arbitrary interference, $G_i = \{1, \ldots, N\}$, so the effective treatments are  
$\{\Z_{\{G_i\}}\}_{\Z\in\calZ} = \{\Z\}_{\Z\in\calZ}$ and so $E_i = 2^N$. Finally, under 1-local interference, $G_i = \Ni_i$, so
$\{\Z_{\{G_i\}}\}_{\Z\in\calZ} = \{\Z_{\Ni_i}\}_{\Z\in\calZ}$, and thus $E_i = 2^{|\Ni_i|}$ under BD.\\

\noindent Recall that $S_i(\Z) = |\mathcal{Z}_i(\Z)|$, where: 
\begin{equation*}
\mathcal{Z}_i(\Z) = \{ \Z' : \Z'_{G_i} = \Z_{G_i} \quad \mbox{and} \quad \mathbb{P}(\Z') > 0\}
\end{equation*}
Under the Bernoulli Design and 1-local interference $G_i = \Ni_i$. So if we denote by 
$\Z_{-\Ni_i}$ the sub-vector of $\Z$ excluding the neighborhood of $i$, $\Z_{-\Ni_i}$ is
of dimension $N-\Ni_i$ and 
\begin{equation*}
	|\mathcal{Z}_i(\Z)| = |\{\Z_{-\Ni_i}\}_\Z| = 2^{N-\Ni_i} = \frac{2^N}{E_i}
\end{equation*}
it immediately follows that $F_i = \frac{1}{E_i}$.

\paragraph{$\sexpg(E_i)$ and $\sexpg(S_i(\Z))$ under BD:}

We have:

\begin{flalign*}
	\sexpg(S_i(\Z)) &= E_\G(2^{N - |\Ni_i|}) \\
	&= 2^{N-1} E_\G\bigg( (1/2)^{d_i}\bigg) \\
	&= 2^{N-1} (p/2 + (1-p))^{N-1}\\
	&= 2^{N-1} (1 - \frac{p}{2})^{N-1}
\end{flalign*}
and so for $p = 1/N$ we have:
\begin{flalign*}
	E_\G(|\calZ_i(Y_i(\A))|) &= 2^{N-1} e^{(N-1)log(1-\frac{1}{2N})}\\
	&\underset{N\rightarrow \infty}{=} 2^{N-1}e^{-\frac{N-1}{2N}}\\
	&\underset{N\rightarrow \infty}{=} \frac{2^{N-1}}{e^{1/2}}
\end{flalign*}
and so the fraction of all assignments that correspond to effective
treatment of $i$ is asymptotically:
\begin{equation*}
	\frac{E_\G(|\calZ_i(Y_i(\A))|)}{|\mathcal{Z}|} = \frac{1}{2e^{1/2}}
\end{equation*}
in the dense case however, we have:
\begin{flalign*}
	E_\G(|\calZ_i(Y_i(\A))|) &= 2^{N-1} e^{(N-1)log(1-\frac{1}{2\sqrt{N}})}\\
	&\underset{N\rightarrow \infty}{=} 2^{N-1}e^{-\frac{N-1}{2\sqrt{N}}}\\
	&\underset{N\rightarrow \infty}{=} \frac{2^{N-1}}{e^{\sqrt{N}/2}}
\end{flalign*}
and so asymptotically:
\begin{equation*}
	\frac{E_\G(|\calZ_i(Y_i(\A))|)}{|\mathcal{Z}|} = \frac{1}{2e^{\sqrt{N}/2}} \rightarrow 0
\end{equation*}

Similarly, we have:
\begin{flalign*}
\sexpg(E_i) &= \sexpg(2^{|\Ni_i|})\\
&= 2\sexpg(2^{d_i}) \\
&= 2 (1 + p)^{N-1}
\end{flalign*}
and so for $p=1/N$, we have:
\begin{flalign*}
\sexpg(E_i) &= 2e^{(N-1) log(1+1/N)} \\
&\rightarrow 2e^{(N-1) / N} \\
&\rightarrow 2e
\end{flalign*}
while for $p=1 / \sqrt{N}$ we have:
\begin{flalign*}
\sexpg(E_i) &=  2e^{(N-1) log(1+1/\sqrt{N})}\\
&\rightarrow 2e^{(N-1) / \sqrt{N}}\\
&\rightarrow \infty
\end{flalign*}

\end{document}